\documentclass{aamas2014}

\usepackage[acm,theorems,greekit]{pamas}
\usepackage[hyphens]{url}
\usepackage[hidelinks]{hyperref}

\usepackage{tikz,pgflibraryshapes}

\usepackage{latexsym,optparams}
\usepackage{array,xspace,multirow,hhline,graphicx,xcolor,tikz,colortbl,tabularx,booktabs,fixltx2e}
\usepackage{verbatim}
\usepackage{fancyvrb}
\usepackage{relsize}

\usepackage[numbers]{natbib}
\def\newblock{\hskip .11em plus .33em minus .07em}

\usepackage{ifthen}
\usepackage{enumerate}
\usepackage{comment}
\usepackage{ragged2e}

\usepackage[vlined]{algorithm2e}

\usepackage{caption} 
\usepackage{xspace}
\usepackage{amsmath,amssymb,mathtools}
\usepackage{tikz}
\usetikzlibrary{shapes}
\usetikzlibrary{positioning}

\usepackage{nicefrac}

\usepackage{subfig}
\usepackage{booktabs}
\usepackage{multirow}

\usepackage{etoolbox}
\patchcmd{\maketitle}{\@copyrightspace}{}{}{}

\newcommand{\N}{\mathbb{N}}
\newcommand{\C}{\mathcal{C}}

\newcommand{\dimc}{\operatorname{dim}_C}

\newcommand{\checkconddim}[1][k]{{\sc Check-Condorcet-Dimension-${#1}$}\xspace}
\newcommand{\satcheckconddim}[1][k]{{\sc SAT-Check-Condorcet-Dimension-${#1}$}\xspace}

\newlength{\wordlength}

\title{Finding Preference Profiles of Condorcet Dimension k\\ via SAT}
\numberofauthors{1}
\author{
	\alignauthor
		Christian Geist \\
		\affaddr{Technische Universit\"at M\"unchen} \\
		\affaddr{Munich, Germany} \\
		\email{geist@in.tum.de}
}

\begin{document}

\maketitle

\begin{abstract}
Condorcet winning sets are a set-valued generalization of the well-known concept of a Condorcet winner. 
As supersets of Condorcet winning sets are always Condorcet winning sets themselves, an interesting property of preference profiles is the size of the smallest Condorcet winning set they admit. 
This smallest size is called the Condorcet dimension of a preference profile. 
Since little is known about profiles that have a certain Condorcet dimension, we show in this paper how the problem of finding a preference profile that has a given Condorcet dimension can be encoded as a satisfiability problem and solved by a SAT solver. 
Initial results include a minimal example of a preference profile of Condorcet dimension $3$, improving previously known examples both in terms of the number of agents as well as alternatives. 
Due to the high complexity of such problems it remains open whether a preference profile of Condorcet dimension $4$ exists. 
\end{abstract}

\section{Introduction} 
\label{sec:introduction}
The contribution of this paper is twofold. Firstly, we provide a practical implementation for finding a preference profile for a given Condorcet dimension by encoding the problem as a boolean satisfiability (SAT) problem \citep{BHMW09a}, which is then solved by a SAT solver. 
This technique has proven useful for a range of other problems in social choice theory \citep[see, \eg ]{TaLi09a,GeEn11a,BGS14a,BrGe15a} and can easily be adapted. 
For instance, only little needs to be altered in order answer similar questions for dominating sets rather than Condorcet winning sets. 
Secondly, we give an answer to an open question by \citet{ELS11a} and provide a minimal example of a preference profile of Condorcet dimension $3$, which we computed using our implementation. 
This profile involves $6$ alternatives and agents only, improving the size of previous examples both in terms of agents and alternatives.\footnote{For instance, the example in \citet{ELS11a} required $15$ alternatives and agents.} 
The formalization in SAT turns out to be efficient enough, not only to discover this particular profile of Condorcet dimension $3$, but also to show its minimality.


\vfill\eject
\section{Preliminaries} 
\label{sec:preliminaries}

Let~$A$ be a set of $m$ alternatives and $N=\{1,\ldots,n\}$ a set of agents. The preferences of agent~$i \in N$ are represented by a linear (\ie reflexive, complete, transitive, and antisymmetric) \emph{preference relation}~$R_i\subseteq A\times A$.
The interpretation of $(a,b) \in R_i$, usually denoted by $a \mathrel{R_i} b$, is  that agent~$i$ values alternative~$a$ at least as much as alternative~$b$. 
A \emph{preference profile} $R = (R_1,\dots, R_n)$ is an $n$-tuple containing a preference relation $R_i$ for each agent $i \in N$.

Let $R$ be a preference profile. 
As introduced by \citet{ELS11a}, we now define the notion of a Condorcet winning set through an underlying covering relation between sets of alternatives and alternatives: 
A set of alternatives $X$ $\theta$-covers an alternative $y$ (short: $X \succ_R^\theta y$) if 
\[
   |\{i\in N \mid \exists x\in X \text{ such that } x \mathrel{R_i} y\}| > \theta n\text{.}
\]

A set of alternatives $X$ is called a \emph{Condorcet winning set} if for each alternative $y\notin X$ the set $X$ $\frac{1}{2}$-covers $y$. 
The set of all Condorcet winning sets of $R$ will be denoted by $\C(R)$. 
The \emph{Condorcet dimension} $\dimc(R)$ is defined as the size of the smallest Condorcet winning set the profile $R$ admits, \ie 
\[
 \dimc(R) := \min\{k\in\N \mid k=|S| \text{ and $S\in\C(R)$}\}\text{.}
\]

\begin{example}
	Consider the preference profile $R$ depicted in \figref{fig:example}. 
	As $R$ does not have a Condorcet winner $\dimc(R)\geq 2$. 
	It can easily be checked that $\{a,b\}$ (like any other two-element set in this case) is a Condorcet winning set of $R$ and, thus, $\dimc(R)=2$.
\end{example}

\begin{figure}[bht] 
	\centering
	$
	\begin{array}{ccc}     
		1&1&1\\\cline{1-3}
	   	a&b&c\\                   
	   	b&c&a\\                   
	   	c&a&b
	\end{array}
	$
	\caption{A preference profile of Condorcet dimension $2$.}
	\label{fig:example}
\end{figure}

In this work, we address the computational problem of finding a preference profile of a given Condorcet dimension. 
To this end, we define the problem of checking whether for a given number of agents $n$ and alternatives $m$ there exists a preference profile $R$ with $\dimc(R)=k$.

\medskip
\noindent
\textbf{Name}: \checkconddim \\
\noindent
\textbf{Instance}: A pair of natural numbers $n$ and $m$.\\
\noindent
\textbf{Question}: Does there exist a preference profile $R$ with $n$ agents and $m$ alternatives that has Condorcet dimension of at least $k$? \\
\smallskip

\begin{table*}[tb]
\small
\centering

\begin{tabular}{lcccccc} 
\toprule
	Preference profiles	& $n=3$		& $n=5$ 			& $n=6$ 		&	$n=7$ &	$n=10$	&	$n=15$						\\ 
\midrule                                                                                                                								
$m=5$	& $\sim 1.7 \cdot 10^{6}$		&  $\sim 2.5 \cdot 10^{10}$	& $\sim 3.0 \cdot 10^{12}$ 	& $\sim 3.6 \cdot 10^{14}$	&	$\sim 6.2 \cdot 10^{20}$	&	$\sim 1.5 \cdot 10^{31}$	\\ 
$m=6$	& $\sim 3.7 \cdot 10^{8}$		& $\sim 1.9 \cdot 10^{14}$	& $\mathbf{\sim 1.4 \cdot 10^{17}}$ 	& $\sim 1.0 \cdot 10^{20}$	 &	$\sim 3.7 \cdot 10^{28}$	&	$\sim 7.2 \cdot 10^{42}$	\\
$m=7$	& $\sim 1.3 \cdot 10^{11}$			& $\sim 3.3 \cdot 10^{18}$	& $\sim 1.6 \cdot 10^{22}$ 	& $\sim 8.3 \cdot 10^{25}$	 &	$\sim 1.1 \cdot 10^{37}$	&	$\sim 3.4 \cdot 10^{55}$	\\
$m=10$	& $\sim 4.8 \cdot 10^{19}$			& $\sim 6.3\cdot 10^{32}$	& $\sim 2.3\cdot 10^{39}$ & $\sim 8.3\cdot 10^{45}$ 	& $\sim 4.0\cdot 10^{65}$	&	$\sim 2.5 \cdot 10^{98}$		\\
\bottomrule

\end{tabular}
\caption{Number of objects involved in the \checkconddim[3] problem. For $k=3$ the subsets of size $2$ are the candidates for Condorcet winning sets.}
\label{tab:magnitudes}
\end{table*}

Note that the following simple observation can be used to prune the search space in terms of the number of alternatives. 

\begin{observation}
	\label{obs:indAlternatives}
	If there is a preference profile $R$ of Condorcet dimension $\dimc(R)$ involving $m$ alternatives, then there is also one of the same dimension involving $m+1$ alternatives.
\end{observation}
\begin{proof}
	Let $R$ be a preference profile on a set of $m$ alternatives $A$ with $\dimc(R)$. 
	We need to construct a preference profile $R'$ on a set of $m+1$ alternatives $A'=A\cup\{a'\}$ with $a'\notin A$ such that $\dimc(R')=\dimc(R)$. 
	For each $i$, define $R'_i:=R_i\cup \{(x,a')\mid x\in A\}$, \ie add $a'$ in the last place of agent $i$'s preference ordering. 
	It is then immediately clear that $\C(R)\subseteq \C(R')$, which establishes $\dimc(R)\geq \dimc(R')$. 
	On the other hand, if we assume $\dimc(R)> \dimc(R')$, then there exist a Condorcet winning set $S'$ for $R'$ of size $k:=|S'|<\dimc(R)$. 
	This set, however, must--by the construction of $R'$--also be a Condorcet winning set for $R$; a contradiction.
\end{proof}


\section{Methodology} 
\label{sec:methodology}

The number of objects potentially involved in the \checkconddim problem are given in \tabref{tab:magnitudes} for $k=3$. 
It is immediately clear that a na\"ive algorithm will not solve the problem in a satisfactory manner.
This section describes our algorithmic efforts to solve this problem for reasonably large instances.

\subsection{Translation to propositional logic (SAT)}
\label{sec:sat}
In order to solve the problem \checkconddim[k] for arbitrary $k\in\N$, we follow a similar approach as \citet{BGS14a}: 
we translate the problem to propositional logic (on a computer) and use state-of-the-art SAT solvers to find a solution. 
At a glance, the overall solving steps are shown in~\algref{alg:SATcheck}.

\begin{algorithm}[t]
	\textbf{Input:} positive integers $n$ and $m$\\
	\textbf{Output:} whether there exists a preference profile $R$ with $n$ agents and $m$ alternatives and $\dimc(R)\geq k$\\
		\tcc{Encoding of problem in CNF}
		File cnfFile\;
		\ForEach{agent $i$}{
			cnfFile += Encoder.reflexivePreferences($i$)\;
			cnfFile += Encoder.completePreferences($i$)\;
			cnfFile += Encoder.transitivePreferences($i$)\;
			cnfFile += Encoder.antisymmetricPreferences($i$)\;
		}
		\ForEach{set $S\subseteq A$ with $|S|=k-1$}{
			cnfFile += Encoder.noCondorcetWinningSet($S$)\;
		}
		\tcc{Symmetry breaking}
		cnfFile += Encoder.neutrality()\;
		\tcc{SAT solving}
		satisfiable = SATsolver.solve(cnfFile)\;
		\If{instance is satisfiable}{
			\Return true\;
		} \Else {
			\Return false
		}
		\caption{\satcheckconddim}
		\label{alg:SATcheck}
\end{algorithm}

Generally speaking, the problem at hand can be understood as the problem of finding a preference profile that satisfies certain conditions---here: having a Condorcet dimension of at least $k$). 
Thus, a satisfying instance of the propositional formula to be designed should represent a preference profile. 
To capture this, a formalization based on two types of variables suffices. 
The boolean variable $r_{i,a,b}$ represents $a\mathrel{R_i} b$, \ie agent $i$ ranking alternative $a$ at least as high as alternative $b$; and the variable $c_{S,y}$ stands for the set $S$ covering alternative $y$. 

In more detail, the following conditions/axioms need to be formalized:\footnote{The further axiom for neutrality is not required for correctness, but speeds up the solving process. It is discussed in \secref{subsec:opt}.} 
\begin{enumerate}
	\item All $n$ agents have linear orders over the $m$ alternatives as their preferences (short: linear preferences)
	\item For each set $S\subseteq A$ with $|S|=k-1$, it is not the case that $S$ is a Condorcet winning set (short: no Condorcet set)
\end{enumerate}

For the first axiom, we encode reflexivity, completeness, transitivity, and anti-symmetry of the relation $R_i$ for all agents $i$. 
The complete translation to CNF (conjunctive normal form, the established standard input format for SAT solvers) is given exemplarily for the case of transitivity; the other axioms are converted analogously.

In formal terms transitivity can be written as 
\begin{eqnarray*}
	&& (\forall i)(\forall x,y,z)\left(x\mathrel{R_i}y \wedge y\mathrel{R_i}z \rightarrow x\mathrel{R_i}z \right) \\
	&\equiv & (\forall i)(\forall x,y,z)\left(r_{i,x,y} \wedge r_{i,y,z} \rightarrow r_{i,x,z} \right) \\
	&\equiv & \bigwedge_i\bigwedge_{x,y,z}\left(\neg \left(r_{i,x,y} \wedge r_{i,y,z}\right) \vee r_{i,x,z} \right) \\
	&\equiv & \bigwedge_i\bigwedge_{x,y,z}\left(\neg r_{i,x,y} \vee \neg r_{i,y,z} \vee r_{i,x,z} \right)\text{,}
\end{eqnarray*}
which then translates to the pseudo code in \algref{alg:transitivity} for generating the CNF file.
The key in the translation of the inherently higher order axioms to propositional logic is (as pointed out by \citet{GeEn11a} already) that because of finite domains, all quantifiers can be replaced by finite conjunctions or disjunctions, respectively.

\begin{algorithm}
	\caption{Encoding of transitivity of individual preferences}
	\label{alg:transitivity}
	\ForEach{agent $i$}{
		\ForEach{alternative $x$}{
			\ForEach{alternative $y$}{
				\ForEach{alternative $z$}{
					variable\_not(r$(i,x,y)$)\;
					variable\_not(r$(i,y,z)$)\;
					variable(r$(i,x,z)$)\;
					newClause\;
				}
			}
		}
	}
\end{algorithm}

In all algorithms, a subroutine r$(i,x,y)$ takes care of the compact enumeration of variables.\footnote{The DIMACS CNF format only allows for integer names of variables. But since we know in advance how many agents and alternatives there are, we can simply use a standard enumeration method for tuples of objects.} 

The axiom ``no Condorcet set'' can be formalized in a similar fashion, but requires further subroutines to avoid an exponential blow-up of the size of the formula in CNF. 
In short, the axiom can be written as
\begin{eqnarray*}
	&& (\forall S\subseteq A) \left( |S|=k-1 \rightarrow S\notin \C(R) \right)\\
	&\equiv & (\forall S\subseteq A) \left( |S|=k-1 \rightarrow (\exists y\notin X) S \nsucc_R^\theta y\right) \\
	&\equiv & \bigwedge_{\substack{S\subseteq A\\|S|=k-1}} \bigvee_{y\notin X} \neg c_{S,y} \text{.}
\end{eqnarray*}

It remains as part of this axiom to define a sufficient condition for $S \succ_R^\theta y$. 
In the following, we denote the smallest number of agents required for a strict $\theta$-majority by $m(n):=\lfloor \theta k \rfloor + 1$. 
In formal terms, we write for each set $S\subseteq A$ with $|S|=k-1$ and each alternative $y\notin X$: 

\begin{eqnarray*}
	&& S \succ_R^\theta y \leftarrow \left((\exists M\subseteq N) |M|=m(n) \wedge \right.\\
	&& \qquad\qquad\qquad \left.(\forall i\in M)(\exists x\in S) x \mathrel{R_i} y \right) \\
	&\equiv & S \succ_R^\theta y \vee \left((\forall M\subseteq N) |M|=m(n) \rightarrow \right.\\
	&& \qquad\qquad\qquad \left.(\exists i\in M)(\forall x\in S) \neg x \mathrel{R_i} y \right) \\
	&\equiv & c_{S,y} \vee \left(\bigwedge_{\substack{M\subseteq N\\|M|=m(n)}} \bigvee_{i\in M} \bigwedge_{x\in S} \neg r_{i,x,y} \right) \text{.}
\end{eqnarray*}

In order to avoid an exponential blow-up when converting this formula to CNF, variable replacement (a standard procedure also known as Tseitin transformation) is applied. In our case, we replaced $\bigwedge_{x\in S} \neg r_{i,x,y}$ by new variables of the form $h_{S,y,i}$ and introduced the following defining clauses:\footnote{Note that one direction of the standard bi-implication suffices here.} 
\begin{eqnarray*}
	&& \bigwedge_{\substack{S\subseteq A\\|S|=k-1}} \bigwedge_{y\in A} \bigwedge_{i\in N} \left(h_{S,y,i} \rightarrow \bigwedge_{x\in S} \neg r_{i,x,y} \right) \\
	&\equiv & \bigwedge_{\substack{S\subseteq A\\|S|=k-1}} \bigwedge_{y\in A} \bigwedge_{i\in N} \left(\neg h_{S,y,i} \vee \bigwedge_{x\in S} \neg r_{i,x,y} \right) \\
	&\equiv & \bigwedge_{\substack{S\subseteq A\\|S|=k-1}} \bigwedge_{y\in A} \bigwedge_{i\in N} \bigwedge_{x\in S} \left(\neg h_{S,y,i} \vee \neg r_{i,x,y} \right) \text{.}
\end{eqnarray*}
In this case, the helper variables even have an intuitive meaning as $h_{S,y,i}$ enforces that for no alternative $x\in S$ it is the case that agent $i$ prefers alternative $y$ over alternative $x$, \ie agent $i$ does not contribute to $S$ $\theta$-covering $y$.

Note that the conditions like $|S|=k-1$ can easily be fulfilled during generation of the corresponding CNF formula on a computer. 
For enumerating all subsets of alternatives of a given size we, for instance, used Gosper's Hack \citep{Knut11a}. 

The corresponding pseudo code for the ``no Condorcet set'' axiom can be found in \algref{alg:noCondSet}.

\begin{algorithm}
	\caption{Encoding of the axiom ``no Condorcet set''}
	\label{alg:noCondSet}
		\ForEach{set $S\subseteq A$ with $|S|=k-1$}{
			\ForEach{alternative $y\notin S$}{
				variable\_not(c$(S,y)$)\;
			}
			newClause\;
			\tcc{Definition of variable $c_{S,y}$}
			\ForEach{set $M\subseteq N$ with $|M|=m(n)$}{
				variable(c$(S,y)$)\;
				\ForEach{agent $i\in M$}{
					variable(h$(S,y,i)$)\;
				}
				newClause\;
			}
			\tcc{Definition of auxiliary variable $h_{S,y,i}$}
			\ForEach{agent $i\in N$}{
				\ForEach{$x\in S$}{
					variable\_not(r$(i,x,y)$)\;
					variable\_not(h$(S,y,i)$)\;
					newClause\;
				}
			}
		}
\end{algorithm}

With all axioms formalized in propositional logic, we are now ready to search for preference profiles $R$ of Condorcet dimension $\dimc(R)\geq k$.
Before we do so, however, we describe a (standard) optimization technique called symmetry breaking, which speeds up the solving process of the SAT solver. 

\subsection{Optimized computation}
\label{subsec:opt}

Observe that from a given example of a preference profile $R$ with $\dimc(R)\geq k$ we can always generate further examples simply by permuting the (names of the) alternatives. 
One could say that all positive witnesses to the \satcheckconddim problem are invariant under permutations of the alternatives. 
Therefore, we implemented a standard technique in SAT solving called symmetry breaking; here in the form of setting agent 1's preferences to a fixed preference ordering, for instance to lexicographic preferences.
This trims the search space for the SAT solver and therefore reduces the runtime of the solving process. 
An encoding can be achieved simply by adding a subformula of the form 
\[
\bigwedge_{x<y} r(n_1, x, y)\text{,}
\]
which sets the first agents preferences to lexicographic ordering. 


\bigskip
\section{Initial Results} 
\label{sec:results}

All computations were run on a Intel Core i5, 2.66GHz (quad-core) machine with 12 GB RAM using the SAT solver \textsc{plingeling} \citep{Bier13a}.

When called with the parameters $n=m=6$, our implementation of \satcheckconddim returns the preference profile $R^{\text{dim3}}$ within about one second. 
$R^{\text{dim3}}$ is a smallest preference profile of Condorcet dimension $3$ and is shown in \figref{fig:dim3}.\footnote{The witnesses for all sets $S\subseteq A$ with $|S|=2$ not being Condorcet winning sets are also returned by \satcheckconddim[3] and can be obtained from the output in \figref{fig:dim3out}. That there is a larger set (\eg $\{a,b,c\}$) which forms a Condorcet winning set can easily be confirmed manually (or by calling \satcheckconddim[4]).} 

Furthermore, it turns out that this preference profile is a smallest profile of Condorcet dimension $3$. 
All strictly smaller profiles (\ie with less agents and at most as many alternatives, or with less alternatives and at most as many agents) can be shown to have a Condorcet dimension of at most 2 via \satcheckconddim[3].\footnote{The running time to check all cases again is only a few seconds.} 

\begin{figure}[tb] 
	\centering
	$
	\begin{array}{cccccc}     
		1&1&1&1&1&1\\\cline{1-6}
	   	a&b&c&d&e&f\\                   
	   	b&c&d&f&a&e\\                   
	   	c&e&f&b&d&a\\ 
		d&d&a&e&b&b\\ 
		e&f&e&a&c&c\\ 
		f&a&b&c&f&d
	\end{array}
	$
	\caption{A smallest preference profile of Condorcet dimension $3$ (with $n=6$ agents $m=6$ alternatives).}
	\label{fig:dim3}
\end{figure}

An overview of further (preliminary) results can be found in \tabref{tab:complete}.

\begin{figure}[h] 
\begin{Verbatim}[frame=single,fontsize=\relsize{-0.5}]
Model (decoding of satisfying assignment) found:
Agent 0: 0 > 1 > 2 > 3 > 4 > 5
Agent 1: 2 > 3 > 5 > 0 > 4 > 1
Agent 2: 5 > 4 > 0 > 1 > 2 > 3
Agent 3: 3 > 5 > 1 > 4 > 0 > 2
Agent 4: 4 > 0 > 3 > 1 > 2 > 5
Agent 5: 1 > 2 > 4 > 3 > 5 > 0
does not have a Condorcet winning set of size 2 
(6 agents and 6 alternatives).
Witnesses:
{0, 1} does not cover alternative(s): 5
{0, 2} does not cover alternative(s): 4
{1, 2} does not cover alternative(s): 0
{0, 3} does not cover alternative(s): 4
{1, 3} does not cover alternative(s): 0
{2, 3} does not cover alternative(s): 0
{0, 4} does not cover alternative(s): 5
{1, 4} does not cover alternative(s): 5
{2, 4} does not cover alternative(s): 1
{3, 4} does not cover alternative(s): 2
{0, 5} does not cover alternative(s): 3
{1, 5} does not cover alternative(s): 3
{2, 5} does not cover alternative(s): 1
{3, 5} does not cover alternative(s): 2
{4, 5} does not cover alternative(s): 3
\end{Verbatim}
\caption{Output of \satcheckconddim[3] for $n=6$ agents and $m=6$ alternatives.}
\label{fig:dim3out}
\end{figure}

\begin{table}[tbh]
	\centering
\begin{tabular}{*{13}{c}}
	\toprule
	$m \backslash n$ 	&	1	&	2	&	3	&	4	&	5	&	6	&	7	&	8	&	9	&	10	&	11	&	12	\\
	\midrule
	1					&	--	&	--	&	--	&	--	&	--	&	--	&	--	&	--	&	--	&	--	&	--	&	--	\\
	2					&	--	&	--	&	--	&	--	&	--	&	--	&	--	&	--	&	--	&	--	&	--	&	--	\\
	3					&	--	&	--	&	--	&	--	&	--	&	--	&	--	&	--	&	--	&	--	&	--	&	--	\\
	4					&	--	&	--	&	--	&	--	&	--	&	--	&	--	&	--	&	--	&	--	&	--	&	--	\\
	5					&	--	&	--	&	--	&	--	&	--	&	--	&	--	&	--	&	--	&	--	&	--	&	--	\\
	6					&	--	&	--	&	--	&	--	&	--	&	+	&	--	&		&	--	&	--	&	--	&	+	\\
	7					&	--	&	--	&	--	&	--	&	--	&	+	&	--	&		&	--	&	+	&		&	+	\\
	8					&	--	&	--	&	--	&		&	--	&	+	&		&		&		&	+	&		&	+	\\
	9					&	--	&	--	&	--	&		&		&	+	&		&		&		&	+	&		&	+	\\
	10					&	--	&	--	&	--	&		&		&	+	&		&		&		&	+	&		&	+	\\
	\bottomrule
\end{tabular}
\caption{Preliminary collection of results obtained with \satcheckconddim[3] for different numbers of alternatives $m$ and voters $n$. A plus (+) stands for a preference profile found; a minus (--) for the fact that all preference profiles have a Condorcet winning set of size $2$.}
\label{tab:complete}
\end{table}


\section{Outlook and future work} 
\label{sec:outlook}
Our implementation might be useful to find preference profiles of Condorcet dimension $4$, a problem that has been raised by \citet{ELS11a}. 
Even though with the current formalization the solving process did not terminate within a reasonable amount of time, we intend to further pursue this direction in future work. 
Adding further symmetry breaking clauses (which make use of anonymity in addition to neutrality) could be a first step in this direction. 

Furthermore, one could extend the notion of Condorcet dimension to other individual preferences, \eg with agents having weak (i.e., ties are allowed) or even incomplete preferences. 
Because of the high flexibility of our SAT formalization, one can easily apply the same method to analyze these related concepts and questions.\footnote{For the two suggested variants, deleting axioms from the formalization suffices.} 

A formalization with other solving techniques, \eg ASP \citep{GKKS12a}, might be another way to achieve the desired performance. 


\section*{Acknowledgments}
This material is based upon work supported by Deutsche Forschungsgemeinschaft under grant {BR~2312/9-1}. The author thanks Felix Brandt and Hans Georg Seedig for helpful discussions and their support.


\end{document}